\documentclass[aps,prl,twocolumn]{revtex4-1}

\usepackage{amsmath,amsfonts,amssymb,amsthm, bbm}

\newcommand\Tr{\mathrm{Tr}}
\newcommand{\id}{\mathbbm 1}
\newcommand{\bsa}{\mathcal B_s(\mathcal H)}
\newcommand{\tsa}{\mathcal T_s(\mathcal H)}
\newcommand{\qeff}{\mathcal E(\mathcal H)}
\newcommand{\ceff}{\mathcal E(\Omega,\Sigma)}
\newcommand{\cstate}{\mathcal S(\Omega,\Sigma)}
\newcommand{\qstate}{\mathcal S(\mathcal H)}
\newcommand{\meas}{(\Omega,\Sigma)}
\newcommand{\smeas}{\mathcal M_{\mathbb R}(\Omega,\Sigma)}
\newcommand{\measf}{\mathcal F_{\mathbb R}(\Omega,\Sigma)}
\newcommand{\bsadual}{\bsa^*}

\newcommand\abs[1]{\left|#1\right|}
\newcommand\av[1]{\left\langle #1 \right\rangle}
\newcommand\norm[1]{\|#1\|}
\newtheorem{corollary}{Corollary}
\newtheorem{theorem}{Theorem}

\usepackage{hyperref} % must be loaded last

\begin{document}

%---------------------------------------------------------------------------------------------------------------------
\title{Necessity of negativity in quantum theory}
\author{Christopher Ferrie, Ryan Morris and Joseph Emerson}
\affiliation{
Institute for Quantum Computing and Department of Applied Mathematics,
University of Waterloo,
Waterloo, Ontario, Canada, N2L 3G1}
%\affiliation{
%Department of Applied Mathematics,
%University of Waterloo,
%Waterloo, Ontario, Canada, N2L 3G1}

\begin{abstract}
A unification of the set of quasi-probability representations using the mathematical theory of frames
was recently developed for quantum systems with finite dimensional Hilbert spaces, in which it was proven that such representations require negative probability in either the states or the effects. In this article we extend those results to Hilbert spaces of infinite dimension, for which the celebrated Wigner function is a special case. Hence, this article presents a unified framework for describing the set of possible quasi-probability representations of quantum theory, and a proof that the presence
of negativity is a necessary feature of such representations.
\end{abstract}

\date{\today}

\maketitle

%---------------------------------------------------------------------------------------------------------------------

%\tableofcontents

%---------------------------------------------------------------------------------------------------------------------
%\section{Introduction\label{S:Introduction}}
%---------------------------------------------------------------------------------------------------------------------
\textbf{Introduction.}  The seeds of the so-called quasi-probability distributions of quantum mechanics where sown in 1932 by Wigner and the now famous function which bears his name.  The Wigner quasi-probability distribution is a function on the classical particle phase space which represents the quantum state of a particle.  In fact, the Wigner function is prescribed by a mapping placing the quantum state space in one-to-one correspondence with (a subset of) functions on phase space.  If not for its negative values, the Wigner function would be a bona fide joint probability distribution of the position and momentum of a quantum particle.

The fact that the Wigner function is but one of infinitely many choices of such quasi-probability distributions suggests the search for a true joint probability distribution representing the quantum state of a particle.  However, it was Wigner himself who proved that such a quest is futile \cite{Wigner1971Quantum}.  Precisely, Wigner showed that a linear map $W$ which assigns to each quantum state $\rho$ a probability density $\mu_\rho(q,p)\geq0$ does not exist if it also satisfies the marginal properties:
\begin{align}
\langle q|\rho|q\rangle&=\int_{\mathbb R} \mu_\rho(q,p) dp,\\
\langle p|\rho|p\rangle&=\int_{\mathbb R} \mu_\rho(q,p) dq.
\end{align}

There is a strong tradition in physics, especially in quantum optics, of considering negativity and related features of the Wigner function and similar quasi-probability representations as a tool to identify strictly non-classical features of quantum states \footnote{See Reference \cite{Kenfack2004Negativity} for a recent review.}. However, it should be understood that any such representation is non-unique and to some extent arbitrary. Moreover, in most, if not all, of these representations the kinematic or ontic space of the representation is presumed to be the usual canonical phase space of classical physics. In the broader context of attempting to represent quantum mechanics as a classical probability theory, the measurable space need not necessarily correspond to the phase space of some classical canonical variables.  Here we consider the problem of whether it is possible to construct a classical reformulation of quantum mechanics, analogous to the one ruled out by Wigner but allowing for an arbitrary measurable space, that is, not restricted to the phase space of a classical physics. More precisely, we are considering the broad class of representation obtained by any injective affine function mapping quantum states to probability measures.  Since we are interested in mapping probabilistic state spaces to each other, we require our mappings to be affine.  Affine maps are those which preserve convex (equivalently, probabilistic) combinations \footnote{Not requiring affinity would lead to contextual ontological models \cite{Harrigan2007Ontological}.}.  We write such a map formally $T:\qstate\to\cstate$.  Here, $\meas$ denotes a measurable space, where $\Sigma$ is a $\sigma$-algebra. Over this space, $\smeas$ denotes the bounded signed measures while $\measf$ denotes the bounded measurable functions.  A signed measure generalizes the usual notion of measure to allow for negative values.  The classical states are the probability measures $\cstate\subset\smeas$.  On the quantum side, $\qstate$ denotes the convex set of quantum states; that is, density operators acting on $\mathcal H$.

Let $\bsa$ denote the bounded self-adjoint operators acting on $\mathcal H$.  It was shown in reference \cite{Stulpe1997Classical} that for every quantum observable $A\in\bsa$, every $\epsilon>0$, and any finite collection of states $\{\rho_i\}\subset\qstate$ there exists a classical observable function $f\in\measf$ such that
\[
\abs{\Tr(\rho_iA)-\int_\Omega f \;d(T\rho_i)}<\epsilon
\]
holds for each $\rho_i$.  The result tells us, remarkably, that the expectation values calculated via the Born Rule are empirically indistinguishable from those calculated using the classical rule.  For finite dimensional quantum systems a stronger result holds \cite{Busch1993Classical}: there exists an affine injective map $T:\mathcal S(\mathbb C^d)\to \mathcal S(\mathbb R^{d^2})$, from density matrices to probability vectors, such that for every Hermitian matrix $A$ there exists a random variable $r\in \mathbb R^{d^2}$ such that
\[
\Tr(\rho A)=\sum_{i=1}^{d^2} r_i (T\rho)_i.
\]
Here, the expectation values calculated via the Born Rule are \emph{exactly} those calculated using the classical rule.

The authors call the mapping $T$ a ``classical representation''.  However, the above representations possess some peculiar features.  The classical effects are the measurable functions taking values in $[0,1]$.  These are denoted $\ceff\subset\measf$ \footnote{These could also be called conditional probabilities or generalized indicator functions in analogy with an ontological models framework \cite{Spekkens2005Contextuality, Harrigan2007Ontological}.}.  The quantum effects, denoted $\qeff$, are the positive operators acting on $\mathcal H$ and bounded above by $\id$.  As shown in reference \cite{Busch1993Classical}, while the quantum observables can be reproduced using the a classical representation, the quantum effects cannot.

In the present work we show that representations constructed as above are not consistent with classical probability theory because they still require a notion of negative probability -- in particular, in their representation of measurements.  Our approach builds on the work of References \cite{Ferrie2008Frame, Ferrie2009Framed}, which develops an approach to representing quantum mechanics that takes into account the requirements for representing both the states and measurements in a self-consistent manner within the framework of classical probability theory. In that work it was shown that, for finite dimensional quantum mechanics, non-negative representations do not exist.

The results in Reference \cite{Ferrie2008Frame, Ferrie2009Framed} were stated only for those quantum systems with associated Hilbert spaces of finite dimension.  While conjectured to be true, the question of whether or not those results held for infinite dimensional systems was left open.  Here we give a more general framework consistent with both finite and infinite dimensional Hilbert spaces.  Precisely, we prove that each quasi-probability representation is equivalent to a frame representation and such representations require negativity.

We first begin with a review of the frame formalism of Reference \cite{Ferrie2008Frame, Ferrie2009Framed}.

%---------------------------------------------------------------------------------------------------------------------
%\section{Frame representations in finite dimensions\label{S:Frame representations in finite dimensions}}
%---------------------------------------------------------------------------------------------------------------------
\textbf{Frame representations in finite dimensions.}  Consider the Hilbert space $\mathbb C^d$ with $d<\infty$ and the set $\Omega$ with $d^2 \leq \abs{\Omega}< \infty$.

A \emph{frame} for the Hilbert space of Hermitian operators is a set of Hermitian operators $\{F_j:j\in\Omega\}$ which satisfy
\begin{equation}\label{frame definition finite dimensions}
a\norm{A}^2\leq \sum_{j\in\Omega} [\Tr(F_j A)]^2 \leq b\norm{A}^2,
\end{equation}
for all Hermitian operators $A$ and some constants $a,b>0$.  Note that if $\abs{\Omega}=d^2$ then a frame is equivalent to a basis.

The key feature of a frame is its \emph{reconstruction formula}: given a frame $F$ any Hermitian operator $A$ can be written
\begin{equation}\label{reconstruction formula finite dim}
A=\sum_{j\in\Omega} \Tr(F_j A) D_j,
\end{equation}
where the set of Hermitian operators $\{D_j:j\in\Omega\}$ is called a \emph{dual frame} of $F$.  When $F$ is a basis, the dual frame is unique.  Otherwise, there are infinitely many choices for a dual.

A classical representation of quantum mechanics is a pair of affine mappings $T:\mathcal S(\mathbb C^d)\to \mathcal S(\Omega)$, $S:\mathcal E(\mathbb C^d)\to\mathcal E(\Omega)$ which satisfy $S(0)=0$ and
\begin{equation}\label{law total prob finite dim}
\Tr(\rho E)=\sum_{j\in\Omega} (T\rho)_j (SE)_j,
\end{equation}
for all $\rho\in \mathcal S(\mathbb C^{d})$ and all $E\in\mathcal E(\mathbb C^{d})$.  Note that, since the classical effects are probabilities, Equation \eqref{law total prob finite dim} is the classical Law of Total Probability.  Unfortunately, a classical representation does not exist \cite{Busch1993Classical,Spekkens2008Negativity, Ferrie2008Frame, Ferrie2009Framed}.  Relaxing only the condition of positivity was proposed as the appropriate definition of a quasi-probability representation.  That is, a quasi-probability representation of quantum mechanics is a pair of affine mappings $T:\mathcal S(\mathbb C^d)\to \mathbb R^\Omega$, $S:\mathcal E(\mathbb C^d)\to\mathbb R^\Omega$ which satisfy $S(0)=0$ and
\begin{align*}
\sum_{j\in\Omega} (T\rho)_j &= 1,\\
\Tr(\rho E)&=\sum_{j\in\Omega} (T\rho)_j (SE)_j,
\end{align*}
for all $\rho\in \mathcal S(\mathbb C^{d})$ and all $E\in\mathcal E(\mathbb C^d)$.

It was shown in \cite{Ferrie2009Framed} that $(T,S)$ is a quasi-probability representation if and only if
\begin{align*}
(T\rho)_j&=\Tr(\rho F_j)\\
(SE)_j&=\Tr(E D_j),
\end{align*}
where $\{F_j:j\in\Omega\}$ is a frame and $\{D_j:j\in\Omega\}$ is one if its duals.  Armed with these results we can begin to create quasi-probability representations of quantum mechanics.  First, we choose one of the many so-called discrete Wigner functions \footnote{See Reference \cite{Ferrie2009Framed} for a review of the popular choices.}. Second, we identify the frame which gives rise to it.  Last, we compute its dual frame to obtain the second half of the quasi-probability representation.

Given a quasi-probability representation $(T,S)$, note that the frame $F$ satisfies
\[
\sum_{j\in\Omega} F_j =\id.
\]
Thus, if the quasi-probability representation satisfies $T\mathcal S(\mathcal C^d)\subset \mathcal S(\Omega)$, the frame is an informationally complete positive operator valued measure (IC-POVM).  Similarly, the dual frame $D$ satisfies
\[
\Tr(D_j)=1,
\]
for all $j\in\Omega$.  Thus, if the the quasi-probability representation satisfies $S\mathcal E(\mathbb C^d)\subset\mathcal E(\Omega)$, the dual frame is a set of density operators.

These definitions and results are tailored to the case $d<\infty$.  Now we will extend them to infinite dimensions answering the questions left open.

%---------------------------------------------------------------------------------------------------------------------
%\section{Frame representations\label{S:Frame representations}}
%---------------------------------------------------------------------------------------------------------------------
\textbf{Frame representations in infinite dimensions.}  For the remainder suppose that the dimension of the Hilbert space $\mathcal H$ is arbitrary and let $\meas$ be a measurable space.  In this section we define the generalization of frame to the space of self-adjoint trace-class operators $\tsa$.

Recall that a frame for Hermitian matrices is equivalent to an informationally complete \emph{operator} valued measure (no positivity required).  If not for the allowance of negative eigenvalues, a frame would be an IC-POVM.  So, we generalize the definition of informationally complete observable for infinite dimensions and define an \emph{operator valued measure} as a map $F:\Sigma\to\bsa$ satisfying $F(\emptyset)=0$, $F(\Omega)=1$ and
\[
F\left(\bigcup_{i=1}^\infty B_i\right)=\sum_{i=1}^\infty B_i,
\]
where the sets $B_i\in\Sigma$ are mutually disjoint and the sum converges in the weak sense.  A \emph{frame} for $\tsa$ is an operator valued measure $F$ for which the map $T:\tsa\to\smeas$,
\[
T(W)(B):=\Tr(WF(B)),
\]
is injective.  The map $T$ is called a \emph{frame representation} of $\tsa$.

Similarly, generalizing the reconstruction formula \eqref{reconstruction formula finite dim} yields the generalized notion of a dual.  That is, given a frame $F$, a \emph{dual frame} to $F$ is a map $D:\Omega\to\bsadual$, the dual space, for which the function
\[
(SA)(\omega):=(D(\omega))(A)
\]
is measurable and satisfies
\begin{equation}\label{reconstruction formula}
A=\int_\Omega SA dF,
\end{equation}
for all $A\in\bsa$ \footnote{We have chosen to generalize the notion of frame via the reconstruction formula.  We are not the first to generalize the familiar definition of frame.  Other generalizations, known by the name \emph{continuous frames}, are defined such that the frame remains an indexed (possibly continuously so) set which satisfies an equation analogous to equation \eqref{frame definition finite dimensions} \cite{Kaiser1990Quantum,Ali1993Continuous}}.

%---------------------------------------------------------------------------------------------------------------------
%\section{Classical and quasi-probability representations\label{S:Classical and quasi-probability representations}}
%---------------------------------------------------------------------------------------------------------------------
\textbf{Classical and quasi-probability representations.}  Now we will generalize the definition of classical representation to the more general measurable space $\meas$.  A \emph{classical representation of quantum mechanics} is a pair of mappings $T:\qstate\to\cstate$ and $S:\qeff\to\ceff$ such that
\begin{enumerate}
\item $T$ and $S$ are affine.
\item $S(0)=0$.
\item For all $\rho\in\qstate$ and $E\in\qeff$,
\begin{equation}\label{LTP in classical rep}
\Tr(\rho E)=\int_\Omega  (SE) d(T\rho).
\end{equation}
\end{enumerate}

As expected, we have
\begin{theorem}\label{T:no classical rep}
A classical representation of quantum mechanics does not exist.
\end{theorem}
\begin{proof}
Suppose $T$ and $S$ form a classical representation of quantum mechanics.  Conditions 1 and 2 imply that $T$ and $S$ can be extended to bounded linear functions $T: \tsa \rightarrow \smeas$ and $S: \bsa \rightarrow \measf$. Thus for all $W \in \tsa$ and $A \in \bsa$, condition 3 gives
\begin{equation}\label{equation in theorem 1}
	\Tr(WA) = \int_\Omega (SA) d(TW) = \Tr(W (T'SA)),
\end{equation}
where $T'$ denotes the dual of the map.  Hence for all $A \in \bsa$, $A = T'SA$.  $S$ must be injective, and so $S^{-1}$ exists, and if $R(S)$ is the range of $S$, then $T'|_{R(S)} = S^{-1}$.  Therefore

\[ T'\ceff \supset T'|_{R(S)}\ceff = S^{-1}\ceff \supset \qeff. \]

Also, if $f \in \ceff$ and $W\in\qstate$, then

\[ 0 \leq \int_\Omega f d(TW) = \Tr(W(T'f))
			\leq \int_\Omega d(TW) = 1, \]
so $T'f \in \qeff$ and $T'\ceff \subseteq \qeff$.  Therefore, $T'\ceff = \qeff$.

However, this is a contradiction since $T'\ceff\subset\qeff$ is a proper inclusion \cite{Stulpe1997Classical}.  A short and intuitive proof that the inclusion is proper was briefly mentioned by Bugajski in Reference \cite{Bugajski1993Classical}.  The proof relies on the notion of \emph{coexistent effects} (See reference \cite{Hellwig1969Coexistent} and references to G. Ludwig therein).  Intuitively, two (classical or quantum) effects are coexistent if they can be measured together.  Any two classical effects are coexistent while two quantum effects are not necessarily coexistent.  Since $T'$ is a linear operator, it preserves the coexistence of effects.  That is, the set $T'\ceff$ is one in which any two elements (quantum effects) are coexistent.  But, again, not all quantum effects can be measured together.  Hence the inclusion is proper.
\end{proof}

In analogy with the finite dimensional case, we generalize the definition of a classical representation to allow for ``negative probabilities''. A \emph{quasi-probability representation of quantum mechanics} is a pair of mappings $T:\qstate\to\smeas$ and $S:\qeff\to\measf$ such that
\begin{enumerate}
\item $T$ and $S$ are affine, $T$ is bounded.
\item $T\rho(\Omega)=1$.
\item $S(0)=0$.
\item For all $\rho\in\qstate$ and $E\in\qeff$,
\begin{equation}\label{LTP in classical rep2}
\Tr(\rho E)=\int_\Omega  (SE) d(T\rho).
\end{equation}
\end{enumerate}

Since a quasi-probability representation in which $T\qstate\subset\cstate$ and $S\qeff\subset\ceff$ is a classical representation, we have immediately from Theorem \ref{T:no classical rep},
\begin{corollary}
A quasi-probability representation of quantum mechanics must have, for some $\rho\in\qstate$, $E\in\qeff$ either $(T\rho)(B)<0$ for some $B\in\Xi$ or $(SE)(\omega) \not\in [0,1]$ for some $\omega\in\Omega$.
\end{corollary}

Note that, strictly speaking, we could have $SE(\omega) > 1$.  However, we still refer to this as negativity since $SE(\omega)$ is meant to be interpreted as a probability and $1-SE(\omega)<0$ should stand on the same footing.  Thus, the above corollary gives us the necessity of negativity in quasi-probability representations.

%---------------------------------------------------------------------------------------------------------------------
%\section{Necessity of frames\label{S:necessity of frames}}
%---------------------------------------------------------------------------------------------------------------------

\textbf{Necessity of frames.}  It is easy to verify that every frame representation defines a quasi-probability representation.  The converse is also true as established by the following theorem:
\begin{theorem}
The pair $(T,S)$ is a quasi-probability representation of quantum mechanics if and only if $T$ is a frame representations of $\tsa$ and $S$ is a dual frame.
\end{theorem}
\begin{proof}
Assume $(T,S)$ is a quasi-probability representation.  Define $F(B):=T'\chi_B$.  By the definition of the dual map
\[
\Tr[\rho F(B)]=\Tr[\rho T'\chi_B]=\int_\Omega \chi_B d(T\rho)=(T\rho)(B).
\]
It is clear then that $F(\emptyset)=0$ and $F(\Omega)=1$ by normalization.  The $\sigma$-additive can be verified by directly substituting the union of an arbitrary sequence of disjoint sets $\{B_i\}\subset \Sigma$.  It is known that $T$ and $S$ can be uniquely extended to bounded linear mappings $T:\tsa\to\smeas$ and $S:\bsa\to\measf$.  Now suppose $TW$ is the zero measure for some $W\in\tsa$.  From the last property of a quasi-probability representation, the Law of Total Probability, we have
\[
\Tr(W A)=\int_\Omega (SA) d(TW)=0.
\]
Since this is true for all $A\in\bsa$, $W=0$.  Therefore $T$ is injective, and we have shown that $T$ is a frame representation.  Now define $D_{\omega}=S^*\delta_\omega$.   So we have
\[
D_{\omega}(A)=S^*\delta_\omega(A)=\int_\Omega (SA) d\delta_\omega=(SA)(\omega).
\]
And, again from the Law of Total Probability,
\[
\Tr[WA]=\int_\Omega (SA) d\Tr[WF]=\Tr\left[W \int_\Omega(SA)dF\right].
\]
Hence
\[
A=\int_\Omega (SA) dF,
\]
and by definition $D$ is dual to $F$.
\end{proof}

\textbf{Conclusion.}  We conclude with an example: the celebrated Wigner function.  We want a joint probability distribution $\mu_\rho(p,q)$ for the state of the quantum system.  From the postulates of quantum mechanics we have a rule for calculating expectation values.  In particular, we can compute the characteristic function
\begin{equation*}\label{Wigner characteristic function}
\phi(\sigma,\mu):=\av{e^{i(\sigma p+\mu q)}}=\Tr(e^{i(\sigma P+\mu Q)}\rho),
\end{equation*}
where $Q$ and $P$ are the position and momentum operators.  Since the characteristic function is just the Fourier transform of the joint probability distribution, we simply invert to obtain
\begin{equation}\label{Wigner Fourier invert of characteristic}
\mu_\rho(p,q)=\frac{1}{(2\pi)^2}\iint_{\mathbb R^2} \Tr(e^{i(\sigma P+\mu Q)}\rho) e^{-i(\sigma p +\mu q)}d\sigma d\mu,
\end{equation}
which is the \emph{Wigner function} of $\rho$ \cite{Wigner1932Quantum}.  The mapping $T:\rho\mapsto\mu_\rho$ is linear and injective and hence we can find a frame $F$ such that
\begin{equation*}\label{Wigner as frame rep}
\mu_\rho(p,q)=\Tr\left[ F(p,q) \rho\right].
\end{equation*}
And, indeed, it is quite easy to see that the frame is
\begin{equation*}\label{Wigner frame}
F(p,q)=\frac{1}{(2\pi)^2}\iint e^{i\sigma (P-p)+i\mu (Q-q)}d\sigma d\mu.
\end{equation*}
This frame has the interesting property of being \emph{tight}: it posses a dual frame which is proportional to itself.  Thus the Wigner function can be inverted via the reconstruction formula
\begin{equation*}\label{Wigner inverse}
\rho = 2\pi\iint_{\mathbb R^2} \mu_\rho(p,q) F(p,q) dpdq.
\end{equation*}
Defining the dual frame as $S:E\mapsto 2\pi\Tr[F(q,p) E]$ gives us a full quasi-probability representation of states and measurements.

In this paper we have shown that the frame representations provide a unification of both finite and infinite dimensional quasi-probability representations.  We have also proven that no classical representation of quantum theory exists or, equivalently, any quasi-probability representation of quantum theory requires negativity.

\textbf{Acknowledgements.}  The authors thank Werner Stulpe for many helpful comments.  This work was financially supported by the government of Canada through NSERC and the Canadian Institute for Advanced Research.

\bibliography{csferrie}
\end{document}